\documentclass[twocolumn,english,prl,floatfix,amsmath,superscriptaddress]{revtex4}
\usepackage[T1]{fontenc}
\usepackage[latin9]{inputenc}
\usepackage{verbatim}
\usepackage{amsmath}
\usepackage{amsthm}
\usepackage{amssymb}

\newtheorem{thm}{Theorem}
\newtheorem{theorem}{Theorem}

\newtheorem{lem}{Lemma}
\newtheorem{prop}{Proposition}

\usepackage{physics}
\usepackage{bbm}
\usepackage{enumitem}
\usepackage{amsfonts}

\makeatother

\providecommand{\ketbra}[1]{\ket{#1}\bra{#1}}
\providecommand{\trace}{\textnormal{tr}}

\newcommand{\ui}{\mathrm{i}}
\newcommand{\ue}{\mathrm{e}}

\def\conv{\textrm{conv}}

\begin{document}

\title{A resource theory of entanglement with a unique multipartite \\
maximally entangled state}

\author{Patricia Contreras-Tejada}

\affiliation{Instituto de Ciencias Matem\'aticas, E-28049 Madrid, Spain}

\author{Carlos Palazuelos}

\affiliation{Departamento de Análisis Matem\'atico y Matem\'atica Aplicada, Universidad
Complutense de Madrid, E-28040 Madrid, Spain}

\affiliation{Instituto de Ciencias Matem\'aticas, E-28049 Madrid, Spain}

\author{Julio I. de Vicente}

\affiliation{Departamento de Matem\'aticas, Universidad Carlos III de Madrid, E-28911,
Legan\'es (Madrid), Spain}
\begin{abstract}
Entanglement theory is formulated as a quantum resource theory in
which the free operations are local operations and classical communication
(LOCC). This defines a partial order among bipartite pure states that
makes it possible to identify a maximally entangled state, which turns
out to be the most relevant state in applications. However, the situation
changes drastically in the multipartite regime. Not only do there
exist inequivalent forms of entanglement forbidding the existence
of a unique maximally entangled state, but recent results have shown
that LOCC induces a trivial ordering: almost all pure entangled multipartite
states are incomparable (i.e.\ LOCC transformations among them are
almost never possible). In order to cope with this problem we consider
alternative resource theories in which we relax the class of LOCC
to operations that do not create entanglement. We consider two possible
theories depending on whether resources correspond to multipartite
entangled or genuinely multipartite entangled (GME) states and we
show that they are both non-trivial: no inequivalent forms of entanglement
exist in them and they induce a meaningful partial order (i.e.\ every
pure state is transformable to more weakly entangled pure states).
Moreover, we prove that the resource theory of GME that we formulate
here has a unique maximally entangled state, the generalized GHZ state,
which can be transformed to any other state by the allowed free operations.
\end{abstract}
\maketitle

\paragraph{Introduction.}

Entanglement is a striking feature of quantum theory with no classical
analogue. Although initially studied to address foundational issues
\cite{entanglement}, the development of quantum information theory
\cite{nielsenchuang} in the last few decades has elevated it to a
resource that allows tasks to be implemented which are impossible
in classical systems. The resource
theory of entanglement \cite{reviews} aims at providing a rigorous
framework to qualify and quantify entanglement and, ultimately,
to understand fully its capabilities and limitations within the realm
of quantum technologies. However, this theory is much more firmly
developed for bipartite than multipartite systems. In fact,
although a few applications have been proposed within the latter setting
such as secret sharing \cite{secretsharing}, the one-way quantum
computer \cite{1wayqc} and metrology \cite{metrology}, a deeper
understanding of the complex structure of multipartite entangled states
might inspire further protocols in quantum information
science and better tools for the study of condensed-matter systems.

The wide applicability of the formulation of entanglement theory as
a resource theory has motivated an active line of work \cite{reviewresource}
that studies different quantum effects from this point of view such
as coherence \cite{coherence}, reference frame alignment \cite{frameness},
thermodynamics \cite{thermodynamics}, non-locality \cite{nonlocality}
or steering \cite{steering}. The main question a resource theory
addresses is to order the set of states and provide means to quantify
their nature as a resource. The so-called free operations are crucial
to this task. This is a subset of transformations, which the
given scenario dictates can be implemented at no cost. Thus, all
states that can be prepared with these operations are free states.
Conversely, non-free states acquire the status of a resource:
granted such states, the limitations of the corresponding scenario
might be overcome. Moreover, the concept of free operations allows
an order relation to be defined. If a state $\rho$ can be transformed
into $\sigma$ by some free operation, then $\rho$ cannot be less
resourceful than $\sigma$ since any task achievable by $\sigma$
is also achievable by $\rho$ as the corresponding transformation
can be freely implemented. However, the converse is not necessarily true. Furthermore, one can introduce resource quantifiers
as functionals that preserve this order.

Since entanglement is a property of systems with many constituents
which may be far away, the natural choice for free operations in this
resource theory is local operations and classical communication (LOCC).
Indeed, parties bound to LOCC can only prepare separable states, and
entangled states become a resource to overcome the constraints imposed
by LOCC manipulation. Nielsen characterized in \cite{nielsen} the
possible LOCC conversions among pure bipartite states, which revealed
that the LOCC ordering reduces to majorization \cite{majorization}
and, remarkably, that there is a unique maximally entangled state
for fixed local dimension. This is because this state can be transformed
by LOCC into any other state of that dimension but no other state
of that dimension can be transformed into it. This state is then regarded
as a gold standard to measure entanglement and, unsurprisingly, it
turns out to be the most useful state for bipartite entanglement applications
such as teleportation. Importantly, the situation changes drastically
in the multipartite case. Here, reference \cite{slocc1} and subsequent
work \cite{slocc2} have shown that there exist inequivalent
forms of entanglement: the state space is divided into classes, the
so-called stochastic LOCC (SLOCC) classes, of states which can be
interconverted 
with non-zero probability by LOCC but
cannot be transformed outside the class by LOCC, even probabilistically.
This in particular shows that no maximally entangled
state can exist for multipartite states. Still, one could in principle study
the ordering induced by LOCC within each SLOCC class. Recent work
\cite{isolation} in this direction has revealed, however, an extreme
feature that culminates with the result of Ref.\ \cite{isolationlast}:
almost all pure states of more than three parties are \textit{isolated},
i.e.\ they cannot be obtained from nor transformed to another inequivalent
pure state of the same local dimensions by LOCC. This means that almost
all pure states are incomparable by LOCC, inducing a trivial
ordering and a meaningless arbitrariness in the construction of entanglement
measures. In this sense, one may say that the resource theory of multipartite
entanglement with LOCC is generically trivial.

We believe this calls for a critical reexamination of the resource
theory of entanglement and, in particular, for the notion of LOCC
as the ordering-defining relation. Indeed, although LOCC transformations
have a clear operational interpretation, this is
not, in fact, the most general class of transformations that maps the set of
separable states into itself. In other words, LOCC is strictly included
in the class of non-entangling operations. Thus, from the abstract point of view of resource theories other consistent theories of entanglement (i.e.\ with separable states being the free states) are possible where the set of free operations is
larger than LOCC. Hence, in principle, these could give a more
meaningful ordering and revealing structure in the set of multipartite
entangled states. To study such possibility is precisely the goal
of this Letter. A similar approach has been taken to address other unsatisfying features of the resource theory of entanglement under LOCC such as irreversibility of state transformations for an arbitrarily large number of copies \cite{irreversibility}. Remarkably, reference \cite{brandaoplenio} has shown that shifting the paradigm from LOCC to asymptotic non-entangling operations provides a reversible theory of asymptotic entanglement interconversion with a unique entanglement measure and this result has been extended in \cite{brandaogour} to arbitrary resource theories under asymptotic resource-non-generating operations \cite{reviewresource}. Also, in the absence of a clear set of physical constraints determining the free operations, certain quantum resource theories have been constructed by first defining the set of free states and then considering classes of operations that preserve this set. This is the case of the resource theory of coherence \cite{coherencePlenio}, which has been found useful in e.g.\ metrology applications \cite{visibility} and quantum channel discrimination \cite{robcoh} and which has subsequently given rise to a fruitful research line considering an operational interpretation for the set of free operations (see \cite{ChiGocoherence,coherence} and references therein).

Since we seek whether a non-trivial theory is at all
possible for single-copy manipulations, we consider here the resource theory of entanglement under the
largest possible class of free operations in this regime: strictly non-entangling operations.
However, multipartite entanglement comes in two different forms. We will call entangled those states that are not fully separable (FS), while we will
call genuinely multipartite entangled (GME) those states which are not biseparable
(BS). Thus, one can formulate two theories: one in which entangled
states are considered a resource and where the free operations are
full separability-preserving (FSP) and the analogous with GME states
and biseparability-preserving (BSP) operations. Interestingly, our
first result is that both formalisms lead to non-trivial theories:
no resource state
is isolated in any of these scenarios. Moreover, we show that there are no inequivalent forms of entanglement. Then, we consider whether there
exists a unique multipartite maximally entangled state in these theories
like in the bipartite case. While we find a negative answer (at least in the simplest non-trivial case of 3-qubit states) for FSP operations, our main result is that the question is
answered affirmatively in the resource theory of
GME under BSP operations. The maximally GME state turns out to be
the generalized Greenberger-Horne-Zeilinger (GHZ) state.

\paragraph{Definitions and preliminaries.}

We will consider $n$-partite systems with local dimension $d$, i.e.\ states
in the Hilbert space $H=H_{1}\otimes\cdots\otimes H_{n}=(\mathbb{C}^{d})^{\otimes n}$.
Given a subset $M$ of $[n]=\{1,\ldots,n\}$ and its complement $\bar{M}$,
we denote by $H_{M}$ the tensor product of the Hilbert spaces corresponding
to the parties in $M$ and analogously with $H_{\bar{M}}$. A pure
state $|\psi\rangle\in H$ is FS (otherwise entangled) if $|\psi\rangle=|\psi_{1}\rangle\otimes|\psi_{2}\rangle\otimes\cdots\otimes|\psi_{n}\rangle$
for some states $|\psi_{i}\rangle\in H_{i}$ $\forall i$, while it
is BS (otherwise GME) if $|\psi\rangle=|\psi_{M}\rangle\otimes|\psi_{\bar{M}}\rangle$
for some states $|\psi_{M}\rangle\in H_{M}$ and $|\psi_{\bar{M}}\rangle\in H_{\bar{M}}$
and $M\subsetneq[n]$. These notions are extended to mixed states
by the convex hull and we define the sets of FS and BS states by
\begin{equation}
\mathcal{FS}=\conv\{\psi:|\psi\rangle\textrm{ is FS}\},\;\mathcal{BS}=\conv\{\psi:|\psi\rangle\textrm{ is BS}\},
\end{equation}
where here and throughout the paper we use the notation $\psi=|\psi\rangle\langle\psi|$
whenever a state is specified as pure. Transformations in quantum
theory are given by completely positive and trace preserving (CPTP)
maps and we say that such a map $\Lambda$ (from and to operators
on $H$) is FSP (BSP) if $\Lambda(\rho)\in\mathcal{FS}$ $\forall\rho\in\mathcal{FS}$
($\Lambda(\rho)\in\mathcal{BS}$ $\forall\rho\in\mathcal{BS}$). We
will say that a functional $E$ taking operators on $H$ to non-negative
real numbers is an FSP-measure (BSP-measure) if $E(\rho)\geq E(\Lambda(\rho))$
for every state $\rho$ and FSP (BSP) map $\Lambda$. This is completely
analogous to entanglement measures, which are required to be non-increasing
under LOCC maps. Although LOCC is a strict subset of the FSP and BSP
maps, some well-known entanglement measures are still FSP- or BSP-measures
and this will play an important role in assessing which transformations
are possible within the two formalisms that we consider here. Indeed,
measures of the form
\begin{equation}
E_{\mathcal{X}}(\rho)=\inf_{\sigma\in\mathcal{X}}E(\rho||\sigma),\label{measures}
\end{equation}
where $\mathcal{X}$ stands for either $\mathcal{FS}$ or $\mathcal{BS}$,
have the corresponding monotonicity property as long as the distinguishability
measure $E(\rho||\sigma)$ is contractive, i.e.\ $E(\Lambda(\rho)||\Lambda(\sigma))\leq E(\rho||\sigma)$
for every CPTP map $\Lambda$. This includes the relative entropy
of entanglement \cite{ree} for $E(\rho||\sigma)=\tr(\rho\log\rho)-\tr(\rho\log\sigma)$
and the robustness ($R_{\mathcal{X}}$) \cite{robustness} for
\begin{equation}
E(\rho||\sigma)=R(\rho||\sigma)=\min\{s:(\rho+s\sigma)/(1+s)\in\mathcal{X}\}.\label{defrobustness}
\end{equation}
If one uses the fidelity $E(\rho||\sigma)=1-F(\rho||\sigma)=1-\tr^{2}\sqrt{\sqrt{\rho}\sigma\sqrt{\rho}}$,
for pure states Eq.\ (\ref{measures}) boils down to the geometric
measure \cite{geometric}, which we will denote by $G_{\mathcal{X}}$
and which is then seen to be a measure under maps that preserve $\mathcal{X}$.
Notice, however, that, as has been recently shown in the bipartite
case in \cite{beyondlocc}, not all LOCC-measures remain monotonic
under non-entangling maps since the latter formalism allows state
conversions that the former does not. In the following, in
order to understand the ordering of resources induced by these theories,
we study which transformations are possible among pure states under
FSP and BSP maps. However, first one should point out that whenever
there exist maps $\Lambda$ and $\Lambda'$ in the corresponding class
of free operations such that $\Lambda(\psi)=\phi$ and $\Lambda'(\phi)=\psi$,
then the states $\psi$ and $\phi$ are equally resourceful and should
be regarded as equivalent in the corresponding theory. This is moreover
necessary so as to have a well-defined partial order. Hence, although
for simplicity we will talk about properties of states, one should
have in mind that one is actually speaking about equivalence classes.
Specifically, it is known that two pure states are interconvertible
by LOCC if and only if they are related by local unitary transformations
\cite{gingrich}. Interestingly, we will see that the equivalence
classes are wider in the resource theory of GME under BSP. It should
be stressed that, to our knowledge, this is the first time that a
resource theory of GME is formulated. Notice that the restriction
to LOCC can only have FS states as free states. Furthermore, allowing
a strict subset of parties to act jointly and classical communication
does not fit the bill either as $\mathcal{BS}$ is not closed under
these operations.

\paragraph{Non-triviality of the theories.}

Our first two results are valid in both the FSP and BSP regimes. Thus,
following the notation above, the two possible classes
of maps will be referred to as $\mathcal{X}$-preserving.

\begin{theorem}[\textbf{collapse of the SLOCC classes}]\label{noSLOCC}
In a resource theory of entanglement where the free operations are
$\mathcal{X}$-preserving maps, all resource states are interconvertible
with non-zero probability, i.e.\ given any pure $\psi_{1},\psi_{2}\notin\mathcal{X}$,
there exists a completely positive and trace non-increasing $\mathcal{X}$-preserving
map $\Lambda$ such that $\Lambda(\psi_{1})=p\psi_{2}$ with $p\in(0,1]$.
\end{theorem}

\begin{theorem}[\textbf{no isolation}]\label{noisolation} In a resource
theory of entanglement where the free operations are $\mathcal{X}$-preserving
maps, no resource state is isolated, i.e.\ given any pure $\psi_{1}\notin\mathcal{X}$
on $H$, there exists an inequivalent pure $\psi_{2}\notin\mathcal{X}$
on $H$ and a CPTP $\mathcal{X}$-preserving map $\Lambda$ such that
$\Lambda(\psi_{1})=\psi_{2}$. \end{theorem}

The full proof of these two results can be found in \cite{supmat}.
The proof of Theorem \ref{noSLOCC} is based on explicitly constructing a completely positive and trace non-increasing
$\mathcal{X}$-preserving map $\Lambda$ such that $\Lambda(\psi_{1})=p\psi_{2}$
whenever it holds that
\begin{equation}
p\leq\frac{1}{R_{\mathcal{X}}(\psi_{2})}\frac{G_{\mathcal{X}}(\psi_{1})}{1-G_{\mathcal{X}}(\psi_{1})}.\label{eq:p_noSFSPclasses}
\end{equation}
Since it can be guaranteed that $R_{\mathcal{X}}(\psi_{2})>0$ and
$0<G_{\mathcal{X}}(\psi_{1})<1$ when $\psi_{1},\psi_{2}\notin\mathcal{X}$,
there always exists%
{} $p\in(0,1]$ such that Eq.\ (\ref{eq:p_noSFSPclasses}) holds.
Theorem \ref{noisolation} then arises as a corollary as, given any
$\psi_{1}\notin\mathcal{X}$, continuity arguments show
that there always exists an inequivalent $\psi_{2}\notin\mathcal{X}$
with $R_{\mathcal{X}}(\psi_{2})$ small enough so that one can take
$p=1$ in Eq.\ (\ref{eq:p_noSFSPclasses}) and construct a CPTP map.

Theorem \ref{noSLOCC} proves that in our case there are no inequivalent
forms of entanglement. This is in sharp contrast to LOCC where, leaving
aside the case $H=(\mathbb{C}^{2})^{\otimes3}$, the state space splits
into a cumbersome zoology of infinitely many different SLOCC classes
of unrelated entangled states. Theorem \ref{noisolation} provides
the non-triviality of our theories. While almost all states turn out
to be isolated under LOCC \cite{isolationlast}, our classes of free
operations induce a meaningful partial order structure where, as in
the case of bipartite entanglement, every pure state can be transformed
into a more weakly entangled pure state. It is important to mention
that the result of \cite{isolationlast} proves generic isolation
when transformations are restricted among GME states 
with the rank of all $n$ single-particle reduced
density matrices equal to $d$. However, Theorem 2 still holds under
this restriction \cite{supmat}. 

\paragraph{Existence of a maximally resourceful state.}

Theorems \ref{noSLOCC} and \ref{noisolation} show that limitations
of the resource theory of multipartite entanglement under LOCC can
be overcome if one considers FSP or BSP operations instead. These
positive results raise the question of whether the induced structure
is powerful enough to have a unique multipartite maximally entangled
state. If this were so, our theories would point to a relevant
class of states that should be at the heart of the applications of
multipartite entanglement in a similar fashion to the maximally
entangled state in the bipartite case. In order to answer this question,
we provide first an unambiguous definition of a maximally resourceful
state which, on the analogy of the bipartite case, depends on the
number of parties $n$ and local dimension $d$: a state $\psi$ on
$H$ is the maximally resourceful state on $H$ if it can be transformed
by means of the free operations into any other state on $H$ \footnote{Notice that this already implies that there exists no free operation
that transforms an inequivalent state on $H$ into $\psi$.}. We analyze first the case of FSP operations, where we find a negative
answer to the above question.

\begin{theorem}\label{fspth} In the resource theory of entanglement
where the free operations are FSP maps, there exists no maximally
entangled state on $H=(\mathbb{C}^{2})^{\otimes3}$. \end{theorem}

Although the details of the proof are given in \cite{supmat}, we
outline here its structure. First, we use that if a maximally entangled
state in this case existed, it would need to be the W state $|W\rangle=(|001\rangle+|010\rangle+|100\rangle)/\sqrt{3}$.
This is because it has been shown in \cite{Wmaxgeom} that the W state
is the unique state in this Hilbert space that achieves the maximal
possible value of $G_{\mathcal{FS}}$, which we have shown above to
be an FSP-measure. Thus, if there existed a maximally entangled state,
it would be necessary that the W state could be transformed by FSP
into any other state. However, we show that there exists no FSP map
transforming the W state into the GHZ state ($|GHZ(3,2)\rangle$ in
Eq.\ (\ref{ghznd}) below). To verify this last claim, it suffices
to find an FSP-measure $E$ such that $E(GHZ)>E(W)$. However, as
discussed above, not many FSP-measures are known and, as with the geometric measure, it is also known that the relative entropy
of entanglement of the W is larger than that of the GHZ state \cite{relativeentropyWGHZ}.
This leaves us then with the robustness measure $R_{\mathcal{FS}}$,
for which we are able to show that $R_{\mathcal{FS}}(W)=R_{\mathcal{FS}}(GHZ)=2$.
This alone does not forbid that $W\to_{FSP}GHZ$ but from the insight
developed in computing these quantities, an
obstruction to such transformation can be found even though they are
equally robust. It is worth mentioning that, to our knowledge, this
is the first time that the robustness is computed for multipartite
states and we have reasons to conjecture that the W and GHZ states
attain its maximal value on $H$, being the only states that do so.

Theorem \ref{fspth} forbids then the existence of a multipartite
maximally entangled state under FSP in the simplest case of $H=(\mathbb{C}^{2})^{\otimes3}$.
However, it is instructive to compare with the LOCC scenario since
these values of $n$ and $d$ make up the only case where no state
is isolated in the latter formalism (aside from the bipartite
case). We show in \cite{supmat} that the W and GHZ states can be
transformed by FSP operations into states that are not obtainable
from any other 3-qubit states by LOCC. These states might be chosen
to lie in different SLOCC classes, so, additionally, this provides
an explicit example of deterministic FSP conversions among states
in different SLOCC classes.

Finally, we study the resource theory under BSP operations where,
remarkably, we find a unique maximally GME state
for any value of $n$ and $d$, given by the generalized
GHZ state
\begin{equation}
|GHZ(n,d)\rangle=\frac{1}{\sqrt{d}}\sum_{i=1}^{d}|i\rangle^{\otimes n}.\label{ghznd}
\end{equation}
\begin{theorem}\label{bspth} In the resource theory of entanglement
where the free operations are BSP maps, there exists a maximally GME
state on every $H$. Namely, $\forall|\psi\rangle\in(\mathbb{C}^{d})^{\otimes n}$,
there exists a CPTP BSP map $\Lambda$ such that $\Lambda(GHZ(n,d))=\psi$.
\end{theorem}

The complete proof of this result is given in \cite{supmat}. The
main idea behind it is to use again the construction of the proof
of Theorems \ref{noSLOCC} and \ref{noisolation}, that shows that
there is a CPTP BSP map $\Lambda$ such that $\Lambda(GHZ(n,d))=\psi$
if $R_{\mathcal{BS}}(\psi)\leq G_{\mathcal{BS}}(GHZ(n,d))/(1-G_{\mathcal{BS}}(GHZ(n,d)))$
(cf.\ Eq.\ (\ref{eq:p_noSFSPclasses})). However, unlike for
the FS case, $G_{\mathcal{BS}}$ is straightforward to compute \cite{puregbs}
in terms of the Schmidt decomposition across every possible bipartite
splitting of the parties $M|\bar{M}$ ($|\psi\rangle=\sum_{i}\sqrt{\lambda_{i}^{M|\bar{M}}}|i\rangle_{M}|i\rangle_{\bar{M}}$)
as
\begin{equation}
G_{\mathcal{BS}}(\psi)=1-\max_{M\subsetneq[n]}\lambda_{1}^{M|\bar{M}},\label{gbs}
\end{equation}
where $\lambda_{1}^{M|\bar{M}}$ is the largest Schmidt coefficient
of $\psi$ in the corresponding splitting. This immediately shows
that the generalized GHZ state has maximal value of the geometric
measure, $G_{\mathcal{BS}}(GHZ(n,d))=(d-1)/d$. Finally, a simple
estimate shows that $R_{\mathcal{BS}}(\psi)\leq d-1$ $\forall|\psi\rangle\in(\mathbb{C}^{d})^{\otimes n}$,
which leads to the desired result.

It follows from the proof that it suffices to have maximal
$G_{\mathcal{BS}}$ to be convertible to any other state by BSP operations.
Thus, any state fulfilling that $G_{\mathcal{BS}}=(d-1)/d$ must automatically
maximize any other BSP-measure. More importantly, this also shows
that any two states achieving this value of the geometric measure
are deterministically interconvertible by BSP operations and, therefore,
belong to the same GME-equivalence class despite potentially not
being related by local unitary transformations. An example of such
class when $d=2$ are GME graph states for which it is known that
$G_{\mathcal{BS}}=1/2$ \cite{Toth_detectingGME}. Hence, all graph
states including the generalized GHZ state are in the equivalence class
of the maximally GME state in this theory. It is remarkable to find
that this very relevant family of states \cite{graph} in quantum
computation and error correction has this feature in a resource theory
of GME and we believe this is worth further research. Another
previously considered family of states that belongs to this equivalence
class is that of absolutely maximally entangled (AME) states \cite{ame},
which is defined as those states for which all reduced density matrices
are proportional to the identity in the maximum possible
dimensions. It follows from Eq.\ (\ref{gbs}) that for all AME states
it holds that $G_{\mathcal{BS}}=(d-1)/d$ (for those values of $n$
and $d$ for which they exist). Equation (\ref{gbs}) also tells us
that a necessary condition for a state to be in the equivalence class
of the maximally GME state is that all single-particle reduced density
matrices must be proportional to the $d$-dimensional identity. However,
this condition is not sufficient: the state in $(\mathbb{C}^{2})^{\otimes4}$
$|\phi\rangle=\sqrt{p}|\phi^{+}\rangle_{12}|\phi^{+}\rangle_{34}+\sqrt{1-p}|\phi^{-}\rangle_{12}|\phi^{-}\rangle_{34}$
($|\phi^{\pm}\rangle=(|00\rangle\pm|11\rangle)/\sqrt{2}$) is a GME
state (if $p\neq0,1$) with this property but $G_{\mathcal{BS}}(\phi)<1/2$
(if $p\neq1/2$).

\paragraph{Conclusions.}
We have shown that non-trivial (i.e.\ without isolation) resource theories of multipartite entanglement are possible in which moreover inequivalent forms of entanglement do not exist. However, no resource theory of non-full-separability can have a maximally
entangled state for 3-qubit states since this is not possible under FSP transformations, the largest conceivable class of free operations (future work should study whether this no-go result generalizes to other values of $n$ and $d$). On the other hand, the BSP paradigm induces a resource theory of GME with a maximally resourceful state. Given this positive result, it would be interesting to analyze further features of this theory and, in particular, whether an operational grounding to this conceptually satisfying structure can be found. We also note that GME does not fulfill Axiom 1 of \cite{brandaogour}, so it is open whether an asymptotically reversible theory of this resource is possible.
\begin{acknowledgments}
This research was funded by the Spanish MINECO through grant MTM2017-88385-P
and by the Comunidad de Madrid through grant QUITEMAD+CMS2013/ICE-2801.
We also acknowledge funding from SEV-2015-0554-16-3 (PCT and CP),
\textquotedblleft Ram\'on y Cajal program\textquotedblright{} RYC-2012-10449
(CP) and the Spanish MINECO grant MTM2017-84098-P (JIdV).
\end{acknowledgments}


\cleardoublepage1
\onecolumngrid
\setcounter{thm}{0}
\setcounter{lem}{0}
\setcounter{prop}{0}
\setcounter{equation}{0}
\begin{center}
\large{\textbf{Supplemental material}}
\end{center}
\twocolumngrid

\title{A resource theory of entanglement with a unique multipartite \\
maximally entangled state: Supplemental material}

\author{Patricia Contreras-Tejada}

\affiliation{Instituto de Ciencias Matemáticas, E-28049 Madrid, Spain}

\author{Carlos Palazuelos}

\affiliation{Departamento de Análisis Matemático y Matemática Aplicada, Universidad
Complutense de Madrid, E-28040 Madrid, Spain}

\affiliation{Instituto de Ciencias Matemáticas, E-28049 Madrid, Spain}

\author{Julio I. de Vicente}

\affiliation{Departamento de Matemáticas, Universidad Carlos III de Madrid, E-28911,
Legan\'es (Madrid), Spain}

\maketitle
We prove the theorems introduced in the main text. For the reader's
convenience, we provide the necessary definitions and restate the
theorems.

Throughout the proof we will use repeatedly that, if $\rho_{1}$ and
$\rho_{2}$ are density matrices, the map
\begin{equation}
\Lambda(\rho)=\textnormal{tr}(A\rho)\rho_{1}+\textnormal{tr}[(\mathbbm1-A)\rho]\rho_{2}\label{eq:lambdacptp}
\end{equation}
is CPTP if $0\leq A\leq\mathbbm1$ (see e.g. \cite{chitambar_entanglement_2017}).

We say that a functional $E$ taking operators on $H$ to non-negative
real numbers is an FSP-measure (BSP-measure) if $E(\rho)\geq E(\Lambda(\rho))$
for every state $\rho$ and FSP (BSP) map $\Lambda$. Measures of
the form 
\begin{equation}
E_{\mathcal{X}}(\rho)=\inf_{\sigma\in\mathcal{X}}E(\rho||\sigma),\label{measures}
\end{equation}
where $\mathcal{X}$ stands for either $\mathcal{FS}$ or $\mathcal{BS}$,
have the corresponding monotonicity property as long as the distinguishability
measure $E(\rho||\sigma)$ is contractive, i.e.\ $E(\Lambda(\rho)||\Lambda(\sigma))\leq E(\rho||\sigma)$
for every $\rho,\sigma\in H$ and every CPTP map $\Lambda$.

As already explained in the main text, two FSP- (BSP-)measures that
will play a key role in developing the resource theory of FSP (BSP)
operations are the geometric measure and the robustness. For the reader's
convenience, we recall their definitions. The robustness is given
by
\begin{equation}
R_{\mathcal{X}}(\cdot)=\min_{\sigma\in\mathcal{X}}R(\cdot||\sigma)\label{eq:robdefn}
\end{equation}
where
\begin{equation}
R(\rho||\sigma)=\min\left\{ s:\frac{\rho+s\sigma}{1+s}\in\mathcal{X}\right\} \,,\label{eq:relrobdefn}
\end{equation}
and the geometric measure, which we only need to consider here for
pure states, boils down to

\begin{equation}
G_{\mathcal{X}}(\cdot)=1-\left(\max_{\ket{\phi}\in\mathcal{X}}\left|\bra{\phi}\ket{\cdot}\right|\right)^{2}\,.
\end{equation}

\section{Non-triviality of the theories}
\begin{thm}
\label{thm:no-slocc} In a resource theory of entanglement where the
free operations are $\mathcal{X}$-preserving maps, all resource states
are interconvertible with non-zero probability, i.e.\ given any pure
$\psi_{1},\psi_{2}\notin\mathcal{X}$, there exists a completely positive
and trace non-increasing $\mathcal{X}$-preserving map $\Lambda$
such that $\Lambda(\psi_{1})=p\psi_{2}$ with $p\in(0,1]\,$. 
\end{thm}
\begin{proof}
Notice that, since $\psi_{1}\,,\psi_{2}\notin\mathcal{X}$ and both
the geometric measure and the robustness are faithful measures \cite{wei_geometric_2003,vidal_robustness_1999},
$R_{\mathcal{X}}(\psi_{2})\,,\:G_{\mathcal{X}}(\psi_{1})>0$. Also,
$G_{\mathcal{X}}(\psi_{1})<1$ because the fully (bi-)separable states
span the whole Hilbert space. Pick $p\in(0,1]$ such that 
\begin{equation}
p\leq\frac{1}{R_{\mathcal{X}}(\psi_{2})}\frac{G_{\mathcal{X}}(\psi_{1})}{1-G_{\mathcal{X}}(\psi_{1})}\label{eq:p_noSFSPclasses}
\end{equation}
and let 
\begin{equation}
\Lambda(\eta)=p\tr(\psi_{1}\eta)\psi_{2}+\tr\left[(\mathbbm1-\psi_{1})\eta\right]\rho_{\mathcal{\mathcal{X}}}\,.\label{eq:lambdaprob}
\end{equation}
Here $\rho_{\mathcal{\mathcal{X}}}\in\mathcal{X}$ is the state which
gives the corresponding robustness of $\psi_{2}\,,$ i.e., $R_{\mathcal{X}}(\psi_{2})=R(\psi_{2}||\rho_{\mathcal{X}})$---cf.
equation (\ref{eq:relrobdefn}). (Note that $\Lambda$ can be completed
to a CPTP $\mathcal{X}$-preserving map by adding a term of the form
$\Lambda'(\eta)=(1-p)\tr(\psi_{1}\eta)\rho_{\mathcal{\mathcal{X}}}\,.)$
Then $\Lambda(\psi_{1})=p\psi_{2}$ and it remains to be shown that
$\Lambda$ is $\mathcal{X}$-preserving. Let $\sigma\in\mathcal{\mathcal{X}}\,.$
Then 
\begin{equation}
\Lambda(\sigma)\propto\psi_{2}+\frac{1}{p}\left(\frac{1}{\tr(\psi_{1}\sigma)}-1\right)\rho_{\mathcal{\mathcal{X}}}\,,
\end{equation}
so $\Lambda(\sigma)/\tr(\Lambda(\sigma))\in\mathcal{\mathcal{X}}$
iff $\frac{1}{p}\left(\frac{1}{\tr(\psi_{1}\sigma)}-1\right)\geq R_{\mathcal{X}}(\psi_{2})\,.$
But this holds from equation \eqref{eq:p_noSFSPclasses} and using
$\tr(\psi_{1}\sigma)\leq1-G_{\mathcal{X}}(\psi_{1})$ $\forall\sigma\in\mathcal{X}\,.$
\end{proof}
\begin{thm}
\label{thm:nontrivialhierarchy} In a resource theory of entanglement
where the free operations are $\mathcal{X}$-preserving maps, no resource
state is isolated, i.e.\ given any pure $\psi_{1}\notin\mathcal{X}$
on $H$, there exists an inequivalent pure $\psi_{2}\notin\mathcal{X}$
on $H$ and a CPTP $\mathcal{X}$-preserving map $\Lambda$ such that
$\Lambda(\psi_{1})=\psi_{2}$.
\end{thm}
\begin{proof}
Consider the map (\ref{eq:lambdaprob}) from the proof of Theorem
\ref{thm:no-slocc}. This map can be made deterministic if $R_{\mathcal{X}}(\psi_{2})$
is sufficiently smaller than $G_{\mathcal{X}}(\psi_{1})\,.$ Indeed,
if
\begin{equation}
\frac{1}{R_{\mathcal{X}}(\psi_{2})}\frac{G_{\mathcal{X}}(\psi_{1})}{1-G_{\mathcal{X}}(\psi_{1})}>1\,,\label{eq:nontrivialhierarchy}
\end{equation}
then we can pick $p=1$ in the map (\ref{eq:lambdaprob}) so $\Lambda$
is CPTP (see equation (\ref{eq:lambdacptp})). Since robustness is
a continuous function of the input state \cite{vidal_robustness_1999},
it can be arbitrarily close to zero and so there exists $\psi_{2}$
such that the above condition is fulfilled for any $\psi_{1}\,.$
Further, $\psi_{1},\psi_{2}$ are inequivalent if they have different
robustness, but $R(\psi_{2})$ can always be picked to be different
from $R(\psi_{1})$ and still satisfying equation (\ref{eq:nontrivialhierarchy}).
\end{proof}
The generic isolation result proven in Ref. \cite{sauerwein_transformations_2017}
holds when transformations are restricted among GME states fully supported
on $H$ (i.e. such that all $n$ single-particle reduced density matrices
have rank $d$). Importantly, $\psi_{1}$ and $\psi_{2}$ in Theorem
\ref{thm:nontrivialhierarchy} may both be fully supported on $H\,,$
in contrast to the LOCC scenario. It suffices to consider
\begin{equation}
\begin{aligned}\ket{\psi_{2}} & =\sqrt{1-\varepsilon}\ket{0}^{\otimes n}+\sqrt{\frac{\varepsilon}{d-1}}\ket{1}^{\otimes n}\\
 & +\dots+\sqrt{\frac{\varepsilon}{d-1}}\ket{d-1}^{\otimes n}
\end{aligned}
\end{equation}
as an example of a GME fully supported state on $H$ which meets the
requirements for small enough $\varepsilon\,.$

\section{FSP regime}
\begin{thm}
\label{thm:noFSPmes} In the resource theory of entanglement where
the free operations are FSP maps, there exists no maximally entangled
state on $H=(\mathbb{C}^{2})^{\otimes3}$.
\end{thm}
To prove this theorem, it is useful to introduce the following two
lemmas in order to compute the robustness of the W and GHZ states.
\begin{lem}
\label{lem:robGHZ} $R_{\mathcal{FS}}(GHZ)=2\,.$ 
\end{lem}
\begin{proof}
The robustness can be bounded from above from the definition (equations
(\ref{eq:robdefn}), (\ref{eq:relrobdefn})), as any fully separable
state which is a convex combination of the GHZ state with a fully
separable state will give an upper bound to the robustness. Ref. \cite{brandao_quantifying_2005}
provides a dual characterization in terms of entanglement witnesses
which we use to bound the robustness from below:
\begin{equation}
R_{\mathcal{FS}}(\rho)=\max\left\{ 0,-\min_{\mathcal{W}\in\mathcal{M}}\tr(\mathcal{W}\rho)\right\} \,.\label{eq:witnessrobustness}
\end{equation}
 A witness for a state $\rho$ is an operator $\mathcal{W}$ such
that $\tr(\mathcal{W}\sigma)\geq0$ for all $\sigma\in\mathcal{FS}$
and $\tr(\mathcal{W}\rho)<0\,.$ If the witness also satisfies $\tr(\mathcal{W}\sigma)\leq1$
for all $\sigma\in\mathcal{FS}$ (which defines the set $\mathcal{M}$
above), then $-\tr(\mathcal{W}\rho)$ is a lower bound to the robustness.

First, we show $R_{\mathcal{FS}}(GHZ)\leq2\,.$ We will use the following
notation as a means to characterize full separability of certain states
(this is a simplified version of the separability criterion in \cite[\S 2.1]{dur_separability_1999}):
a state of the form 
\begin{equation}
\begin{aligned}\rho(\lambda^{+},\lambda^{-},\lambda)=\\
\lambda^{+}GHZ\,+\, & \lambda^{-}GHZ_{-}\,+\frac{\lambda}{6}\sum_{i=001}^{110}\ketbra{i}\,,
\end{aligned}
\label{eq:DCTform}
\end{equation}
where $\ket{GHZ_{-}}=(\ket{000}-\ket{111})/\sqrt{2}$ and the summation
index $i$ ranges from 001 to 110 in binary, is fully separable iff
\begin{equation}
|\lambda^{+}-\lambda^{-}|\leq\lambda/3\,.\label{eq:critsep}
\end{equation}
We must also have $\lambda^{+}+\lambda^{-}+\lambda=1$ for normalization,
and $\lambda^{\pm},\lambda\geq0$ for $\rho(\lambda^{+},\lambda^{-},\lambda)$
to be positive. Thus, the set of fully separable states of the form
(\ref{eq:DCTform}) is a polytope, and this property will be used
later.

Consider the following state: 
\begin{equation}
\frac{1}{3}\left(GHZ+2\rho\left(0,\frac{1}{4},\frac{3}{4}\right)\right)=\rho\left(\frac{1}{3},\frac{1}{6},\frac{1}{2}\right)\,,\label{eq:robGHZ}
\end{equation}
It is straightforward to check that both $\rho\left(0,\frac{1}{4},\frac{3}{4}\right)$
and $\rho\left(\frac{1}{3},\frac{1}{6},\frac{1}{2}\right)$ satisfy
(\ref{eq:critsep}) with equality, so $R_{\mathcal{FS}}(GHZ)\leq2\,.$

Next, we show $R_{\mathcal{FS}}(GHZ)\geq2\,.$ Let

\begin{equation}
\mathcal{W}=\frac{2}{3}\mathbbm1-\frac{8}{3}GHZ+\frac{4}{3}GHZ_{-}\label{eq:robwitGHZ}
\end{equation}
be a candidate witness for this purpose. To show $0\leq\textnormal{tr}(\mathcal{W}\sigma)\leq1$
for all fully separable states $\sigma$, it is enough to restrict
to states $\sigma$ of the form (\ref{eq:DCTform}), as can be shown
by considering the twirling map $T_{GHZ}$ onto the GHZ-symmetric
subspace. This map is defined in \cite{hayashi_bounds_2006}, but
we will only need the following properties: it is FSP and self-dual,
it maps all states onto states of the form (\ref{eq:DCTform}), i.e.
\begin{equation}
T_{GHZ}(\tau)=\rho(\lambda^{+},\lambda^{-},\lambda)
\end{equation}
for every state $\tau$ on $H$ and for some $\lambda^{\pm},\lambda$
and, moreover, these states are fixed points: $T_{GHZ}(\rho(\lambda^{+},\lambda^{-},\lambda))=\rho(\lambda^{+},\lambda^{-},\lambda)$
for all $\lambda^{\pm},\lambda\,.$ In particular, $T_{GHZ}(GHZ)=GHZ$
and the witness $\mathcal{W}$ in equation \eqref{eq:robwitGHZ} is
such that $T_{GHZ}\,(\mathcal{W})=\mathcal{W}\,,$ and so
\begin{equation}
\tr(\mathcal{W}\sigma)=\tr(T_{GHZ}\,(\mathcal{W})\,\sigma)=\tr(\mathcal{W}\,T_{GHZ}(\sigma))
\end{equation}
holds for any state $\sigma\,.$ Therefore, if $0\leq\textnormal{tr}(\mathcal{W}\sigma)\leq1$
holds for all $\sigma\in\mathcal{FS}$ such that $T_{GHZ}(\sigma)=\sigma\,,$
i.e. those of the form \eqref{eq:DCTform} where (\ref{eq:critsep})
holds \cite{eltschka_entanglement_2012,eltschka_optimal_2013}, then
it is guaranteed to hold for any $\sigma\in\mathcal{FS}\,.$

As the space of fully separable GHZ-symmetric states is a polytope,
it is enough to show that $0\leq\trace(\mathcal{W}\sigma)\leq1$ at
the vertices of the polytope, which are (cf. \cite{eltschka_entanglement_2012,eltschka_optimal_2013}):
\begin{align}
\sigma_{1} & =\rho(0,0,1)\nonumber \\
\sigma_{2} & =\rho\left(0,\frac{1}{4},\frac{3}{4}\right)\nonumber \\
\sigma_{3} & =\rho\left(\frac{1}{2},\frac{1}{2},0\right)\\
\sigma_{4} & =\rho\left(\frac{1}{4},0,\frac{3}{4}\right)\,.\nonumber 
\end{align}
It is straightforward to check that $0\leq\trace(\mathcal{W}\sigma_{j})\leq1$
for all $j=1,...,4\,.$ Since $\trace(\mathcal{W\,}GHZ)=-2<0\,,$
$\mathcal{W}$ is a witness for the GHZ-state that meets the required
condition and so $R_{\mathcal{FS}}(GHZ)\geq2\,.$
\end{proof}
\begin{lem}
\label{lem:robW} $R_{\mathcal{FS}}(W)=2\,.$ 
\end{lem}
\begin{proof}
The strategy is similar to the proof of Lemma \ref{lem:robGHZ}. First,
we prove $R_{\mathcal{FS}}(W)\leq2\,.$ We will show that 
\begin{equation}
\eta=\frac{1}{3}(W+2\tau)\,,\label{eq:robW}
\end{equation}
where 
\begin{equation}
\eta=\frac{9}{16}\ketbra{000}+\frac{3}{16}\ketbra{111}+\frac{1}{16}W+\frac{3}{16}\overline{W}\label{eq:rhoconvcomb}
\end{equation}
and 
\begin{equation}
\tau=\frac{3}{8}\ketbra{000}+\frac{1}{8}\ketbra{111}+\frac{3}{8}W+\frac{1}{8}\overline{W}\label{eq:sigmaconvcomb}
\end{equation}
are both fully separable. Here and in what follows, $\overline{W}$
denotes the qubit-flipped version of the $W$-state, 
\begin{equation}
\ket{\overline{W}}=\frac{1}{\sqrt{3}}\left(\ket{110}+\ket{101}+\ket{011}\right)\,.
\end{equation}
As shown in Theorem 6.2 of \cite{eckert_quantum_2002}, if a symmetric
3-qubit state remains positive after partial transposition (PPT),
then it is FS. Since both $\eta$ and $\tau$ are symmetric 3-qubit
states, it is enough to check that they are PPT, which is readily
done, to conclude that they are fully separable.

Another way to see this is by writing $\eta$ and $\tau$ as a convex
combination of fully separable states using a result from \cite{hayashi_entanglement_2008}.
Observe that
\begin{equation}
\begin{aligned}\eta= & \frac{5}{9}\ketbra{000}\\
+ & \frac{4}{9}\left(\frac{1}{2^{6}}\ketbra{000}+\frac{27}{2^{6}}\ketbra{111}\right.\\
 & \left.+\frac{9}{2^{6}}W+\frac{27}{2^{6}}\overline{W}\right)
\end{aligned}
\label{eq:rhoconvcomb-1}
\end{equation}
and 
\begin{equation}
\begin{aligned}\tau= & \frac{1}{9}\ketbra{111}\\
+ & \frac{8}{9}\left(\frac{27}{2^{6}}\ketbra{000}+\frac{1}{2^{6}}\ketbra{111}\right.\\
 & \left.+\frac{27}{2^{6}}W+\frac{9}{2^{6}}\overline{W}\right)
\end{aligned}
\label{eq:sigmaconvcomb-1}
\end{equation}
where, in each case, the first term is clearly fully separable. As
we shall see, the second term is of the form

\begin{equation}
\begin{aligned}\tr & (\phi^{\otimes3}\ketbra{000})\ketbra{000}\\
+ & \tr(\phi^{\otimes3}\ketbra{111})\ketbra{111}\\
+ & \tr(\phi^{\otimes3}\,W\,)\,W\\
+ & \tr(\phi^{\otimes3}\,\overline{W}\,)\,\overline{W}
\end{aligned}
\label{eq:twirled}
\end{equation}
for some qubit state $\phi$. Ref. \cite{hayashi_entanglement_2008}
shows that all states of this form are fully separable. Writing
\begin{equation}
\ket{\phi}=\cos\alpha\ket{0}+\ue^{\ui\beta}\sin\alpha\ket{1}\,.
\end{equation}
and inserting it into equation \eqref{eq:twirled}, the parameter
$\beta$ cancels in all terms and the state in equation \eqref{eq:twirled}
can be written in terms of $\alpha$ alone with $\alpha=\pi/3$ for
$\eta$ and $\alpha=\pi/6$ for $\tau\,.$

Next, we prove $R_{\mathcal{FS}}(W)\geq2\,.$ We will show that 
\begin{equation}
\begin{aligned}A= & \ketbra{000}-3\,W+\\
 & \ketbra{001}+\ketbra{010}+\ketbra{100}+3\,\overline{W}
\end{aligned}
\end{equation}
is a witness for the state $\ketbra{W}$ such that 
\begin{equation}
\tr(AW)=-2
\end{equation}
and 
\begin{equation}
0\leq\tr(A\sigma)\leq1\label{eq:trace1separable}
\end{equation}
for all $\sigma\in\mathcal{FS}\,.$ %

Let $\sigma\in\mathcal{FS}\,.$ Without loss of generality, to prove
(\ref{eq:trace1separable}) we can assume $\sigma=\ketbra{\psi}$
is pure. So we want to show
\begin{equation}
0\leq\tr(A\ketbra{\psi})\leq1\,.\label{eq:trApsi01}
\end{equation}
Notice that $A$ is permutationally invariant, and that we can express
$A$ in the basis of Pauli matrices as
\begin{equation}
A=\sum_{ijk\in{x,y,z}}\lambda_{ijk}\sigma_{i}\otimes\sigma_{j}\otimes\sigma_{k}+\frac{\mathbbm1_{8}}{2}
\end{equation}
for some $\lambda_{ijk}\in\mathbb{R}$ and where $\mathbbm1_{d}$
is the $d$-dimensional identity, so that
\begin{equation}
A'=A-\frac{\mathbbm1_{8}}{2}
\end{equation}
has no identity component in the basis of Pauli matrices. That is,
$A'$ contains only full correlation terms, and it is still permutationally
invariant so it satisfies the conditions of Corollary 5 (ii) in \cite{hubener_geometric_2009}.
In particular, $A'$ can be viewed as a symmetric three-linear form
acting on $\mathbb{R}^{3}$. This means that
\begin{equation}
\max_{\ket{\psi}\in\mathcal{FS}}\left|\tr(A'\ketbra{\psi})\right|
\end{equation}
can be attained by a symmetric state $\ket{\psi}=\ket{a}\ket{a}\ket{a}\equiv\ket{aaa}\,.$
The qubit $\ket{a}$ can be expressed in terms of two real parameters
as
\begin{equation}
\ket{a}=\cos\alpha\ket{0}+\ue^{\ui\beta}\sin\alpha\ket{1}
\end{equation}
and so
\begin{equation}
\left|\tr(A'\ketbra{aaa})\right|=\frac{1}{2}\left|\cos6\alpha\right|\leq\frac{1}{2}\,.
\end{equation}
But this completes the proof, since, by linearity, to show
\begin{equation}
-\frac{1}{2}\leq\tr(A'\ketbra{\psi})\leq\frac{1}{2}
\end{equation}
(which is equivalent to \eqref{eq:trApsi01}) it suffices to show
\begin{equation}
\max_{\ket{\psi}\in\mathcal{FS}}\left|\tr(A'\ketbra{\psi})\right|\leq\frac{1}{2}\,.
\end{equation}
This can be seen by viewing $\tr(A'\ketbra{\psi})$ as a symmetric
three-linear form in $\mathbb{R}^{3}\,.$ If the maximum absolute
value is attained by some state $\ket{a^{*}}\,,$ then the state $\ket{\tilde{a}^{*}}$
which flips the sign of the vector which the three-linear form acts
on will give a minimum of the expression equal to minus the maximum.%
{} Hence,
\begin{equation}
\begin{aligned}\max_{\ket{\psi}\in\mathcal{FS}}\left|\tr(A'\ketbra{\psi})\right| & =\max_{\ket{\psi}\in\mathcal{FS}}\tr(A'\ketbra{\psi})\\
 & =-\min_{\ket{\psi}\in\mathcal{FS}}\tr(A'\ketbra{\psi})\,.
\end{aligned}
\end{equation}
Therefore \eqref{eq:trApsi01} holds true and hence the witness $A$
gives the stated lower bound for the FS robustness of $W$.
\end{proof}
We note that the values obtained for the robustness $R_{\mathcal{FS}}$
of the W and GHZ states show that, unlike in the bipartite case, the
robustness can be strictly larger than the generalized robustness.
The generalized robustness, $R_{G}(\cdot)\,,$ is defined as
\begin{equation}
R_{G}(\cdot)=\min_{\tau\in H}R(\cdot||\tau)
\end{equation}
where, this time, $\tau$ may be separable or entangled. Hence $R_{G}(\cdot)\leq R(\cdot)$
but, in addition, it was shown in \cite{harrow_robustness_2003} that
$R_{G}(\cdot)=R(\cdot)$ for bipartite pure states. However, the generalized
robustness of the W state has been computed in \cite{hayashi_entanglement_2008}
to be $5/4,$ and that of the GHZ state was shown to be $1$ in \cite{hayashi_bounds_2006},
so they are both strictly less than the robustness of these states.
To the best of our knowledge, this is the first time that states such
that $R_{G}(\cdot)<R(\cdot)$ have been found.

We are now ready to prove Theorem \ref{thm:noFSPmes}.
\begin{proof}
As we outlined in the main text, the only candidate for a maximally
entangled state of three qubits is the W state, as it is the unique
state on $H=\left(\mathbb{C}^{2}\right)^{\otimes3}$ that achieves
the maximum value of the FSP-measure $G_{\mathcal{FS}}$ (among both
pure and mixed states, since the convex-roof extension of $G_{\mathcal{FS}}$
to mixed states ensures that the maximum value will always be achieved
by a pure state). So, if there existed a maximally entangled state,
it would need to be possible that the W state be transformed into
any other state via an FSP map. We will assume that there exists an
FSP map $\Lambda$ such that $\Lambda(W)=GHZ\,,$ and will arrive
at a contradiction by showing that there exists a state $\eta\in\mathcal{FS}$
such that $\Lambda(\eta)\not\in\mathcal{FS}\,.$

Let $\Lambda$ be an FSP map such that $\Lambda(W)=GHZ$ and let
\begin{equation}
\eta=\frac{1}{3}W+\frac{2}{3}\tau\in\mathcal{FS}\,,\label{eq:R(W)upper}
\end{equation}
where $\tau,\eta\in\mathcal{FS}\,,$ be the convex combination that
gives the upper bound to $R_{\mathcal{FS}}(W)$ in equations (\eqref{eq:robW}-\eqref{eq:sigmaconvcomb}).
Let $T_{GHZ}$ be the twirling map onto the $GHZ$-symmetric subspace
(defined in \cite{hayashi_entanglement_2008}; see also the proof
of Lemma \ref{lem:robGHZ}). Then, 
\begin{equation}
\begin{aligned}\eta'= & T_{GHZ}\left(\Lambda\left(\frac{1}{3}W+\frac{2}{3}\tau\right)\right)\\
= & \frac{1}{3}GHZ+\frac{2}{3}T_{GHZ}(\Lambda(\tau))\,.
\end{aligned}
\end{equation}
Since both $T_{GHZ}$ and $\Lambda$ are full separability-preserving,
it is the case that $\eta',\Lambda(\tau)\,,$ $T_{GHZ}(\Lambda(\tau))\in\mathcal{FS}\,.$
Now, recall that $\tau$ has a non-zero $W$ component: 
\[
\tau=p\ketbra{W}+(1-p)\xi
\]
for some $p\in(0,1)$ and some state $\xi$, so that 
\begin{equation}
\eta'=\frac{1}{3}GHZ+\frac{2}{3}\left[p\,GHZ+(1-p)T_{GHZ}(\Lambda(\xi))\right]\,.\label{eq:tghzlambda}
\end{equation}
But, as we shall now show, the FS $GHZ$-symmetric state $\upsilon$
such that 
\begin{equation}
\frac{1}{3}GHZ+\frac{2}{3}\upsilon\in\mathcal{FS}\label{eq:GHZ_fs}
\end{equation}
is unique, i.e. if equation (\ref{eq:GHZ_fs}) holds then necessarily
$\upsilon=\rho(0,1/4,3/4)$ as in equation (\ref{eq:robGHZ}). However,
the state appearing in equation (\ref{eq:tghzlambda}) is not $\upsilon$
(since $\tr(\upsilon\,GHZ)=0$) hence, contrary to our assumption,
$\eta'$ cannot be FS.

Recall, from the proof of Lemma \ref{lem:robGHZ} (equation (\ref{eq:DCTform})),
that all $GHZ$-symmetric states are of the form

\begin{equation}
\begin{aligned}\begin{aligned}\rho(\lambda^{+},\lambda^{-},\lambda)=\\
\lambda^{+}GHZ\,+\, & \lambda^{-}GHZ_{-}\,+\frac{\lambda}{6}\sum_{i=001}^{110}\ketbra{i}
\end{aligned}
\end{aligned}
\label{eq:DCTform-1}
\end{equation}
so that equation (\ref{eq:GHZ_fs}) can be expressed in terms of the
$\lambda$ parameters as 
\begin{equation}
\frac{1}{3}GHZ+\frac{2}{3}\rho(\lambda^{+},\lambda^{-},\lambda)=\rho(\frac{1}{3}+\frac{2}{3}\lambda^{+},\frac{2}{3}\lambda^{-},\frac{2}{3}\lambda)\,.\label{eq:GHZ_fs_param}
\end{equation}
States of the form (\ref{eq:DCTform-1}) are fully separable iff 
\begin{equation}
|\lambda^{+}-\lambda^{-}|\leq\lambda/3\,.\label{eq:fs_condition}
\end{equation}
Since this condition must hold for both states $\rho(\cdot,\cdot,\cdot)$
in equation (\ref{eq:GHZ_fs_param}), we must also have 
\begin{equation}
\left|\frac{1}{3}+\frac{2}{3}\lambda^{+}-\frac{2}{3}\lambda^{-}\right|\leq\frac{2}{9}\lambda
\end{equation}
and, for normalisation, we need 
\begin{equation}
\lambda^{+}+\lambda^{-}+\lambda=1\,.
\end{equation}
It is straightforward to check that these three conditions hold only
if 
\begin{equation}
\lambda^{+}=0;\:\lambda^{-}=1/4;\:\lambda=3/4\,,
\end{equation}
which corresponds to the state $\upsilon$ as claimed above.

Therefore $\eta$ in equation (\ref{eq:R(W)upper}) is fully separable,
yet $\eta'=T_{GHZ}(\Lambda(\eta))$ is not fully separable. So $\Lambda$
is not FSP and hence the theorem is proven.
\end{proof}
Thus, Theorem \ref{thm:noFSPmes} forbids the existence of a multipartite
maximally entangled state under FSP in the simplest case of $H=(\mathbb{C}^{2})^{\otimes3}$.
Still, the LOCC case provides a useful comparison since, in this formalism,
the case $n=3,$ $d=2$ is the only one in which no state is isolated
(in addition to $n=2$). Whenever no single maximally entangled state
exists one needs to consider a maximally entangled set (MES) \cite{de_vicente_maximally_2013},
defined as the minimal set of states on $H$ such that any state on
$H$ can be obtained by means of the free operations from a state
in this set. The MES under LOCC for $n=3$ and $d=2$ has been characterized
in \cite{de_vicente_maximally_2013}, and it is found to be relatively
small in the sense that it has measure zero on $H$ (in contrast,
for other values of $n$ and $d$ the fact that isolation is generic
imposes that the MES has full measure on $H$). However, interestingly,
the MES under FSP is smaller even in this case, given that it is strictly
included in the MES under LOCC. This is because, as we will now show,
the W and GHZ states can be transformed by FSP operations into inequivalent
states that are in the MES under LOCC. It is worth mentioning that
the target states may be chosen to lie in different SLOCC classes
with respect to the initial states, and so this gives an explicit
example of deterministic FSP conversions among states in different
SLOCC classes.

Let $\psi_{GHZ}^{+}$ denote states of the form 
\begin{equation}
\ket{\psi_{GHZ}^{+}}=\sqrt{K}\left(\ket{000}+\ket{\phi_{A}\phi_{B}\phi_{C}}\right)\label{eq:ghzplus}
\end{equation}
where 
\begin{equation}
\begin{aligned}\ket{\phi_{A}}= & \cos\alpha\ket{0}+\sin\alpha\ket{1}\,,\\
\ket{\phi_{B}}= & \cos\beta\ket{0}+\sin\beta\ket{1}\,,\\
\ket{\phi_{C}}= & \cos\gamma\ket{0}+\sin\gamma\ket{1}\,,
\end{aligned}
\end{equation}
$\alpha,\beta,\gamma\in(0,\pi/2]$ and $K=(2(1+\cos\alpha\cos\beta\cos\gamma))^{-1}$
is a normalisation factor. States of the form $\psi_{GHZ}^{+}$ are
in the MES under LOCC, since they cannot be reached by any LOCC map
regardless of the input state on $H=(\mathbb{C}^{2})^{\otimes3}$
\cite{dur_three_2000,turgut_deterministic_2010,de_vicente_maximally_2013}.
So the following proposition does not hold in the LOCC regime. 
\begin{prop}
There exists an FSP map $\Lambda$ such that $\Lambda(W)=\psi_{GHZ}^{+}$
for some state of the form $\psi_{GHZ}^{+}\,.$ 
\end{prop}
\begin{proof}
Let 
\begin{equation}
\Lambda(\eta)=\tr(W\eta)\psi_{GHZ}^{+}+\tr\left[(\mathbbm1-W)\eta\right]\tau_{\mathcal{FS}},
\end{equation}
where $\tau_{\mathcal{FS}}\in\mathcal{FS}$ is the state that gives
the robustness of the state $\psi_{GHZ}^{+}\,.$ Clearly, $\Lambda(W)=\psi_{GHZ}^{+}$
and it remains to be shown that $\Lambda$ is FSP. As argued in Theorems
\ref{thm:no-slocc} and \ref{thm:nontrivialhierarchy}, this happens
when
\begin{equation}
R_{\mathcal{FS}}(\psi_{GHZ}^{+})\leq\frac{G_{\mathcal{FS}}(W)}{1-G_{\mathcal{FS}}(W)}=\frac{5}{4}\,.
\end{equation}
But, by continuity of the robustness, such a state $\psi_{GHZ}^{+}$
can always be found by picking the parameters $\alpha,\beta,\gamma$
sufficiently close to zero since in this case the states $\psi_{GHZ}^{+}$
approach the set of FS states.

Anyway, for the sake of completeness, we provide an explicit quantitative
upper bound in what follows. Consider the invertible local operations
\begin{align}
A & =\left(\begin{array}{cc}
1 & \cos\alpha\\
0 & \sin\alpha
\end{array}\right)\,,\nonumber \\
B & =\left(\begin{array}{cc}
1 & \cos\beta\\
0 & \sin\beta
\end{array}\right)\,,\\
C & =\left(\begin{array}{cc}
1 & \cos\gamma\\
0 & \sin\gamma
\end{array}\right)\,.\nonumber 
\end{align}
Applying these to the FS states in equation \eqref{eq:robGHZ} used
to bound the robustness of the GHZ state, 
\begin{equation}
A\otimes B\otimes C\left(\frac{1}{3}GHZ+\frac{2}{3}\upsilon\right)A^{\dagger}\otimes B^{\dagger}\otimes C^{\dagger}\,,
\end{equation}
gives a state proportional to
\begin{equation}
\frac{1}{3}\left(1+\cos\alpha\cos\beta\cos\gamma\right)\psi_{GHZ}^{+}+\frac{2}{3}\frac{4-\cos\alpha\cos\beta\cos\gamma}{4}\upsilon'\,,\label{eq:robpsighzplus}
\end{equation}
where $\upsilon'=A\otimes B\otimes C\upsilon A^{\dagger}\otimes B^{\dagger}\otimes C^{\dagger}$
is still fully separable since local operations cannot create entanglement.
For the same reason, the state in equation (\ref{eq:robpsighzplus})
is fully separable, and hence the robustness of the state $\psi_{GHZ}^{+}$
cannot exceed \footnote{This shows, in particular, that the robustness of all $\psi_{GHZ}^{+}$
states is always less than or equal to 2.} 
\begin{equation}
R_{\mathcal{FS}}(\psi_{GHZ}^{+})\leq\frac{4-\cos\alpha\cos\beta\cos\gamma}{2(1+\cos\alpha\cos\beta\cos\gamma)}\,.
\end{equation}
Clearly, there exist $\alpha,\beta,\gamma\in(0,\pi/2]$ such that
this bound is lower than or equal to 5/4, as required. For an example,
take $\alpha=\beta=\pi/2$ and $\gamma$ such that $\cos\gamma\geq6/7\,.$
\end{proof}
We will now show the converse result: there are FSP maps which take
the $GHZ$-state to states in the $W$-class which are in the MES
under LOCC. Such states are of the form 
\begin{equation}
\begin{aligned}\ket{\psi_{W}}= & \sqrt{x_{1}}\ket{001}+\sqrt{x_{2}}\ket{010}+\sqrt{x_{3}}\ket{100}\end{aligned}
\label{eq:wform}
\end{equation}
where $x_{1}+x_{2}+x_{3}=1$. They are in the MES under LOCC, as no
LOCC map can reach these states for any input state on $H=(\mathbb{C}^{2})^{\otimes3}$
\cite{dur_three_2000,kintas_transformations_2010,de_vicente_maximally_2013},
but (as we will now prove) not under FSP. 
\begin{prop}
There exists an FSP map $\Lambda$ such that $\Lambda(GHZ)=\psi_{W}$
for some state of the form $\psi_{W}\,.$ 
\end{prop}
\begin{proof}
Since $G_{\mathcal{FS}}(GHZ)=1/2\,,$ it suffices to find a state
$\psi_{W}$ such that $R_{\mathcal{FS}}(\psi_{W})\leq1\,,$ which
can be done since the robustness is continuous and there are states
$\psi_{W}$ arbitrarily close to the set of FS states. Then,
\begin{equation}
\Lambda(\eta)=\tr(GHZ\eta)\psi_{W}+\tr\left[(\mathbbm1-GHZ)\eta\right]\tau_{\mathcal{FS}},
\end{equation}
where $\tau_{\mathcal{FS}}\in\mathcal{FS}$ is the state such that
that $R_{\mathcal{FS}}(\psi_{W})=R(\psi_{W}||\tau_{\mathcal{FS}})\,,$
is the required map.
\end{proof}

\section{BSP regime}
\begin{thm}
In the resource theory of entanglement where the free operations are
BSP maps, there exists a maximally GME state on every $H$. Namely,
$\forall|\psi\rangle\in(\mathbb{C}^{d})^{\otimes n}$, there exists
a CPTP BSP map $\Lambda$ such that $\Lambda(GHZ(n,d))=\psi$. 
\end{thm}
\begin{proof}
For every given $\ket{\psi}\in\left(\mathbb{C}^{d}\right)^{\otimes n}\,,$
let 
\begin{equation}
\Lambda(\eta)=\tr(\eta\,GHZ(n,d))\:\psi\,+\tr\left[(\mathbbm1-GHZ(n,d))\eta\right]\rho_{\mathcal{BS}}
\end{equation}
where $\rho_{\mathcal{BS}}\in\mathcal{BS}$ is the state which gives
the (biseparable) robustness of $\psi$ (i.e. $R_{\mathcal{BS}}(\psi)=R(\psi||\rho_{\mathcal{BS}})$).
Then, $\Lambda(GHZ(n,d))=\psi$ and it remains to be shown that $\Lambda$
is BSP. As argued in the proofs of Theorems \ref{thm:no-slocc} and
\ref{thm:nontrivialhierarchy}, this happens iff

\begin{equation}
R_{\mathcal{BS}}(\psi)\leq\frac{G_{\mathcal{BS}}(GHZ(n,d))}{1-G_{\mathcal{BS}}(GHZ(n,d))}\,.\label{eq:EgeomRobustnessBSP}
\end{equation}
As discussed in the main text, it follows from equation (6) that $G_{\mathcal{BS}}(GHZ(n,d))=(d-1)/d\:$
and, therefore, $\Lambda$ is BSP iff $R_{\mathcal{BS}}(\psi)\leq d-1\,.$
It is shown in \cite{vidal_robustness_1999} that for every bipartite
pure state $\psi_{A|B}$ with Schmidt decomposition
\begin{equation}
\psi_{A|B}=\sum_{i}\sqrt{\lambda_{i}^{A|B}}\ket{i}_{A}\ket{i}_{B}
\end{equation}
it holds that
\begin{equation}
R_{\mathcal{BS}}(\psi)=\left(\sum_{i}\sqrt{\lambda_{i}^{A|B}}\right)^{2}-1\,.
\end{equation}
Thus
\begin{equation}
\begin{aligned}R_{\mathcal{BS}}(\psi) & \leq\min_{M\subsetneq[n]}\left(\sum_{i}\sqrt{\lambda_{i}^{M|\bar{M}}}\right)^{2}-1\\
 & \leq d-1
\end{aligned}
\end{equation}
where the latter inequality follows from considering the state with
all eigenvalues $\lambda_{i}=1/d\,.$ Hence, $\forall\ket{\psi}\in\left(\mathbb{C}^{d}\right)^{\otimes n}$
there exists a BSP map $\Lambda$ such that $\Lambda(GHZ)=\psi\,.$
\end{proof}
%

\end{document}